\documentclass[letterpaper]{sig-alternate-2013}

\permission{\copyright 2017 International World Wide Web Conference Committee (IW3C2), 
published under Creative Commons CC BY 4.0 License.}
\conferenceinfo{WWW 2017,}{April 3--7, 2017, Perth, Australia.}
\copyrightetc{ACM \the\acmcopyr}
\crdata{978-1-4503-4913-0/17/04. \\
http://dx.doi.org/10.1145/3038912.3052565  \\
\includegraphics{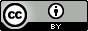}
}
 
\usepackage{algorithm}
\usepackage{algpseudocode}
\usepackage{tikz}

\graphicspath{{img/}}

\newtheorem{theorem}{Theorem}

\newcommand{\prob}[1]{\mathbb{P}\left[ #1 \right]}

\newcommand{\const}{\ensuremath{\mathrm{const}}}
\newcommand{\erdosrenyi}{Erd\H{o}s-R\'enyi }
\newcommand{\erdosandrenyi}{Erd\H{o}s and R\'enyi}
\newcommand{\rdag}[2]{\mbox{\textsf{rdag}$(#1, #2)$}}

\newcommand{\naivecitemany}[2]{#1~et~al.~\cite{#2}}
\newcommand{\naiveciteone}[2]{#1~\cite{#2}}

\begin{document}

\title{Why Do Cascade Sizes Follow a Power-Law?}

%
%
%
%
%

\numberofauthors{4} 
%
\author{
%
%
\alignauthor
Karol W\k{e}grzycki\\
       \affaddr{Institute of Informatics, University of Warsaw, Poland}\\
       \email{k.wegrzycki@mimuw.edu.pl}
\alignauthor
Piotr Sankowski\\
       \affaddr{Institute of Informatics, University of Warsaw, Poland}\\
       \email{sank@mimuw.edu.pl}
\and
\alignauthor
Andrzej Pacuk\\
       \affaddr{Institute of Informatics, University of Warsaw, Poland}\\
       \email{apacuk@mimuw.edu.pl}
\alignauthor
Piotr Wygocki\\
       \affaddr{Institute of Informatics, University of Warsaw, Poland}\\
       \email{wygos@mimuw.edu.pl}
}

\maketitle
\begin{abstract}

    We introduce \emph{random directed acyclic graph} and use it
    to model the information diffusion network. Subsequently, we
    analyze the \emph{cascade generation model} (CGM) introduced by Leskovec et
    al.~\cite{leskovec-blog}. Until now only empirical studies of this model
    were done. In this paper, we present the first theoretical proof that the
    sizes of cascades generated by the CGM follow the power-law distribution,
    which is consistent with multiple empirical analysis of the large social
    networks. We compared the assumptions of our model with the Twitter social
    network and tested the goodness of approximation.

\end{abstract}


\section{Introduction}

Each day billions of instant messages, comments, articles, blog posts, emails,
tweets and other various mediums of communication are exchanged in the
reciprocal, social relations. The study of the information propagation through
network is more and more demanded. Such models of propagation are used to
minimize transmission costs, enhance the security and prevent information leaks or
predict a propagation of malicious software among the users~\cite{moro}.

When considering state-of-the-art models of information diffusion, the
underlying network structure of a transmission is constructed based on the known
connections (e.g., the graph of followers in the Twitter network). Here,
we have discovered that the graph of
information dissemination has noteworthy features, unexploited in the previous
works. It has been well known that more active individuals in the network have
more acquaintances~\cite{seismic}.  We have conducted experiments
confirming those observations in the information diffusion network and showed 
its underlying structure.

Our model of the information diffusion network explains power-law (or Pareto)
distribution of the number of informed nodes (cascade size). This is a major
improvement over the state-of-the-art models, which give predictions inconsistent with
the real data~\cite{cikm2014}. Random power-law graphs may sufficiently
describe the follower-followee relation in the social network~\cite{brach}, but
these graphs may not necessarily characterize the medium of an information
propagation.  The results of our study will allow researchers to enhance their
models of the information transmission and will enable them to develop a
framework to validate these models.
\subsection{Related Work}

Up to the best of our knowledge this work presents the first theoretical analysis of the cascade
size distribution using \emph{cascade generation} to model the spread of the information in
the social networks.  The first experimental analysis has been conducted by
\naivecitemany{Leskovec}{leskovec-blog}.
They proposed a cascade generation model and simulated it on the dataset of blog
links. Since then, the research on the cascade sizes has become a fruitful
field (for more references see~\cite{rogers-diffusion}).

Through study on an epidemiology and a solid-state physics, different models
such as SIR (susceptible\--infectious\--rec\-over\-ed) or SIS
(susceptible\--infectious\--susceptible) has been employed to model the dynamics
of spread of the information. However, all of these models assume that everyone
in the population is in contact with everyone else~\cite{bayley}, which is
unrealistic in large social networks.

The classical example of the modified spreading process incorporates the effect
of a \emph{stifler}~\cite{wlosi}. \emph{Stiflers} never spread the information
even if they were exposed to it multiple times. Nevertheless, \emph{stiflers}
can actively convert other spreaders or susceptible nodes into \emph{stiflers}.
That complicated logic may lead to the elimination of the epidemic
threshold~\footnote{Epidemic threshold determines whether the global epidemic occurs or the
disease simply dies out.} and has been actively developed~\cite{brach2}.

In 2002 \naiveciteone{Watts}{watts} proposed exact solution of
the global cascade sizes on an arbitrary random graph. Notwithstanding, this
process of the information propagation called the threshold model is utterly different
from \emph{cascade generation} by \naivecitemany{Leskovec}{leskovec-blog} and does not fully
explain the dynamics of modern social networks like the Twitter.

\naivecitemany{Iribarren}{moro} have developed the similar model, where an
integro-differential equations have been introduced. That equations describe the
cascade sizes when the number of messages send by a node is described by the
Harris discrete distribution.However, the general solution to their equations is
not known and merely solutions for nontrivial cases (e.g., superexponential
processes~\cite{moro}) has been considered.  Our discoveries provide much
simpler method and lucidly explain, that the underlying social graph is far more
complex than just the graph of followers.

Results similar to ours were also obtained in the study on the bias of traceroute
sampling. \naivecitemany{Achlioptas}{traceroute1} characterize the degree
distribution of a BFS tree for a random graph with a given degree distribution.
Their explanation why the degree distribution under traceroute sampling exhibits
power-law motivated researchers to study bias in P2P systems~\cite{traceroute3}
and network discovery~\cite{traceroute4}. Their research also resulted in the
development of the new tools in the social networks sampling~\cite{traceroute2}.

In the seminal paper of~\naivecitemany{Leskovec}{leskovec-dvd}, the cascade size distribution in the
network of recommendation has been analyzed. \naivecitemany{Leskovec}{leskovec-dvd} showed that the product
purchases follow a ``long tail'' distribution, where a significant fraction of
sold items were rarely sold items. They fit the data to the power-law
distribution and discovered, that the parameters may differ for distinct networks
(remarkably, the power-law exponent was close to $-1$ for DVD recommendation
cascades).

One of the greatest issue in analyzing the cascade size distribution is the lack of
a good theoretical background for this process. 
Recently, researchers~\cite{hypertext2016} has presented new models, designed
to fit real distribution of cascade sizes. Moreover, goodness of fit test
against CGM model has been conducted. It turns out, that CGM works well for 
small cascades, however it is unreliable for the large ones. The introduction of
time dependent parameters significantly improves predictions for the large
cascades~\cite{hypertext2016}.

In future, theoretical studies on cascades size distribution could explain
phenomena in Gossip-based routing~\cite{iwanicki}, rumor virality and influence
maximization~\cite{max-influence}, recurrence of the cascades~\cite{recur} or
assist in forecasting rumors~\cite{forecast}. Currently the research in
this area is purely empirical and we need more theoretical models to understand
the process of information propagation~\cite{hypertext2016}.

\section{Modeling Information Cascades} \label{model}

Intuitively, information cascades are generated by the following process: one individual
passes the information to all its acquaintances. Then in each round newly
informed nodes randomly decide to pass it to their acquaintances. This process continues
until no new individuals are informed. The graph generated by the spread of the
information is called the cascade.

The \emph{cascade generation model} (CGM) established by~\cite{leskovec-blog}
introduces a single parameter $\alpha$ that measures how infectious a passed
information is. More precisely, $\alpha$ is the probability that the
information will be passed to the acquaintance.

According to \naivecitemany{Leskovec}{leskovec-blog}, the cascade is generated
by the following:

\begin{enumerate}
    \item Uniformly at random pick a starting point of the cascade
        and add it to the set of \emph{newly informed} nodes.
    \item Every \emph{newly informed} node independently with
        the probability $\alpha$ informs their direct neighbors.
    \item Let \emph{newly informed} be the set of nodes that has
        been informed for the first time in step $2$ and add them
        to the generated cascade.
    \item Repeat steps $2$ and $3$ until \emph{newly informed} set
        is empty.
\end{enumerate}

In this model we assume that all nodes have an identical impact ($\alpha = \const$)
on their neighbors and all generated cascades are trees, since
we pick a single, initial node. It is not a
major problem, since the most of cascades are trees~\cite{leskovec-blog}.

\begin{figure}[!ht]
    \centering
    \includegraphics[width=0.5\textwidth]{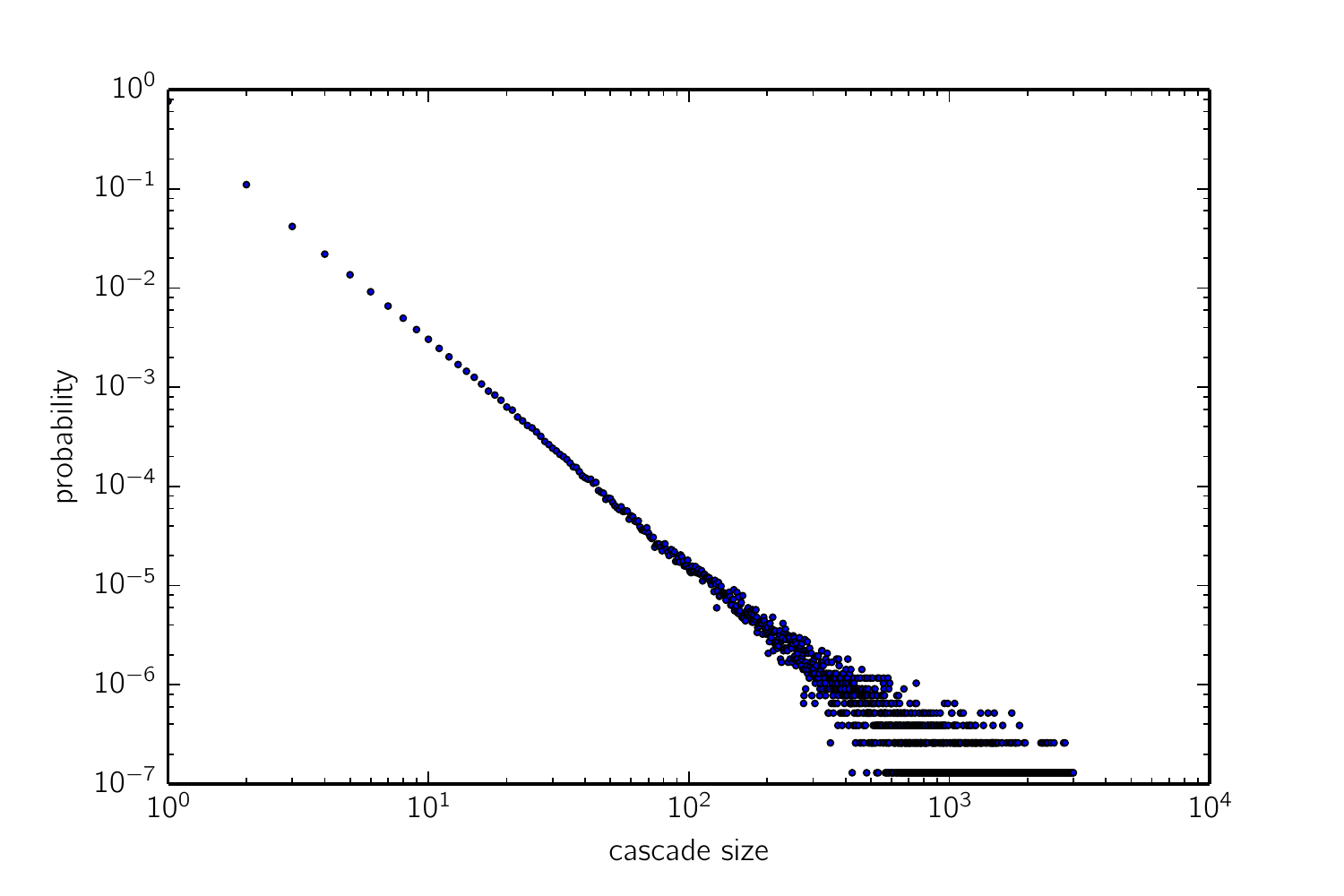}
    \caption{The log-log plot of the cascade size distribution on the Twitter dataset.
        It follows a power-law with an exponent $-2.3$.}
     \label{fig:twitter_reach_distribution}
\end{figure}

\begin{figure}[!ht]
    \centering
    \includegraphics[width=0.5\textwidth]{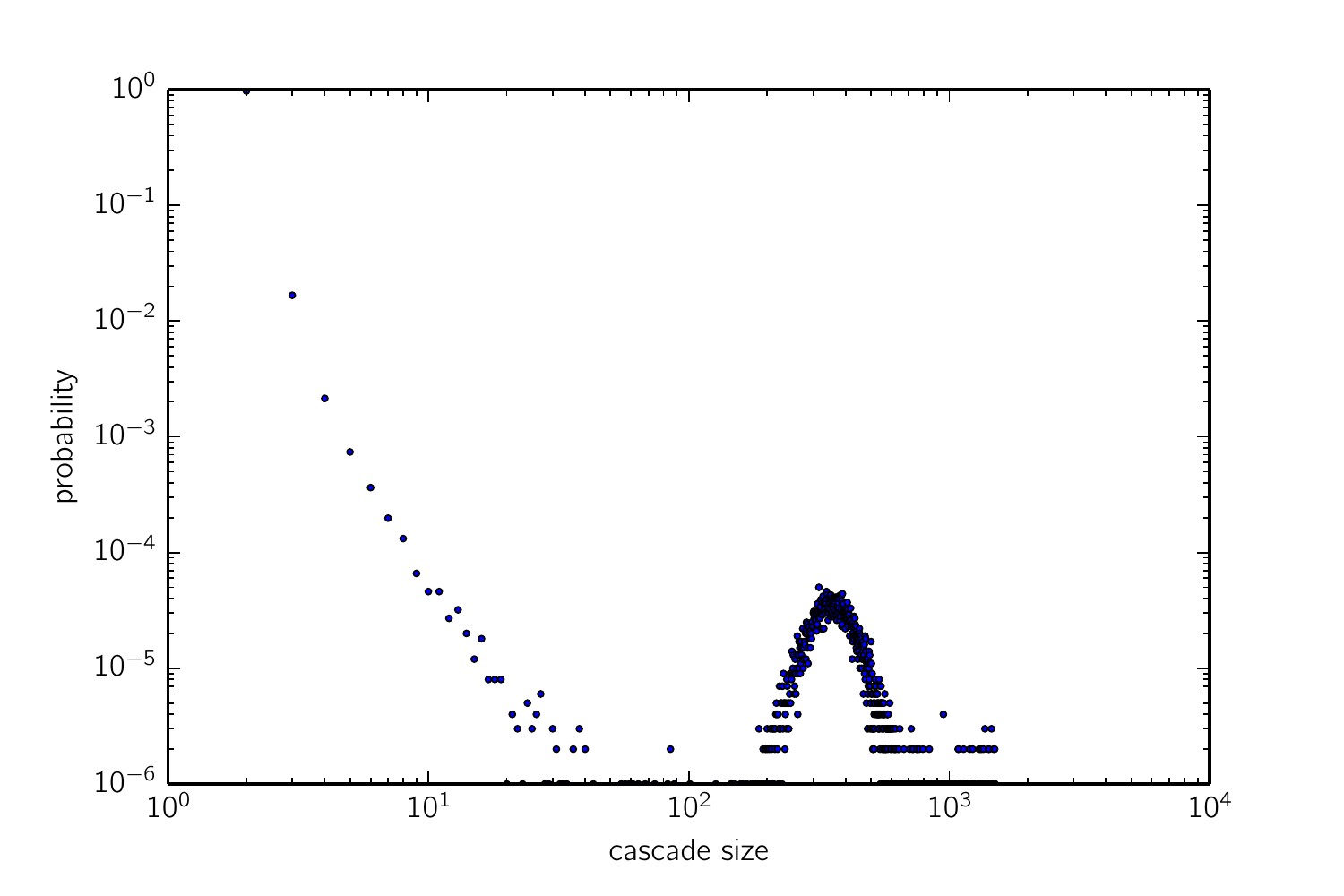}
    \caption{The log-log plot of the cascade size distribution predicted by the simulation
        of CGM using the Twitter followers graph.
        Note the phase transition, absent in real
        data.}
     \label{fig:model_const_distribution}
\end{figure}

Recently, information diffusion models have been evaluated on large social
networks and it has been observed~\cite{greg,cikm2014} that the actual cascade
size distribution is inconsistent with the distribution predicted by the model.
First, the simulation registered an obvious phase transition (the cascades are
either extremely large or are smaller than $100$) (see
Figure~\ref{fig:model_const_distribution}). No such gap has been observed in the
real data (see Figure~\ref{fig:twitter_reach_distribution}). Second, the
probability of the large cascade is intolerably high in
simulations.~\footnote{According to~\cite{cikm2014} the probability that the
cascade will have a size greater than $2\,500$ is $0.00008$ for the real data
and $0.0134$ for a simulation.}

There has been many attempts to readjust the \emph{cascade generation
model}~\cite{ghosh, cikm2014, lerman}, but as we have discussed above, those
attempts oversimplify the process or incorrectly describe the distribution of
cascades sizes. We introduce novel observations concerning the underlying social
network and based on them we introduce a theoretical model of information
propagation.

\subsection{Information Transmission} \label{information-transmission}

The underlying network of rumor spreading is unknown. Even if we would have
the network of all social contacts, the rumor could propagate through the mass
media with a random interaction or even evolve in time. Because of that, when
analyzing the social media we focus on the network of information propagation.
The same technique has been used by \naivecitemany{Leskovec}{leskovec-blog},
but their network was generated by links in blogs. To observe a non-trivial structure of
the social network we considered actions generated by replying to messages. Such
replies in the Twitter microblogging network are called \emph{retweets}. We
analyzed a set of over 500
millions tweets and the retweets from a $10\%$ sample of all tweets from May 19 to May 30 2013.
Each retweet contains identifiers of the cited and replying users. Based on that, we
generate a directed graph of the information transmission (same as~\cite{seismic,cikm2014}).
We use data published in~\cite{hypertext2016,twitter-data} and
publish our code and experimental results on~\cite{rdag-code}.

It is well known that the degree distribution of the generated graph follows a
power-law~\cite{brach}. However, this characterization does not necessarily
describe the network of the information transmission. The intuition is that the
nodes with greater degree are more active. Hence, when information spreads
through the underlying network the distribution of spreading nodes should
prefer nodes with higher degree, in consequence inflating the probability for these
nodes.

Moreover, we have observed the hierarchical structure of the graph of retweets.
The probability that a popular blogger replies to the message of an unpopular
one is extremely small.

To confirm our intuition, we have determined the distribution of
neighboring degrees for all nodes with a given degree. As shown on
Figure~\ref{fig:twitter_retweets_followers_degrees} each degree has a distinct
distribution of neighbors' degrees (implementation is available on~\cite{rdag-code}).
It means that it is more likely that a node is followed by some popular nodes when node itself is
popular. Moreover, there is a pattern:
the probability decreases with degree (for followers with degree
greater than the followed node). Based on this observation we will
model the aforementioned distribution as an approximated step function. This
observation is consistent with the state-of-the-art analysis~\cite{leskovec-blog}
and the most of cascades are ``tree-like'' or ``stars''. According to
our knowledge no further research has been conducted for examining distributions of
neighbors degrees for node with a given degree (only cumulative degree for
every cascade has been studied~\cite{leskovec-blog}).

\begin{figure}[!ht]
    \centering
    \includegraphics[width=0.5\textwidth]{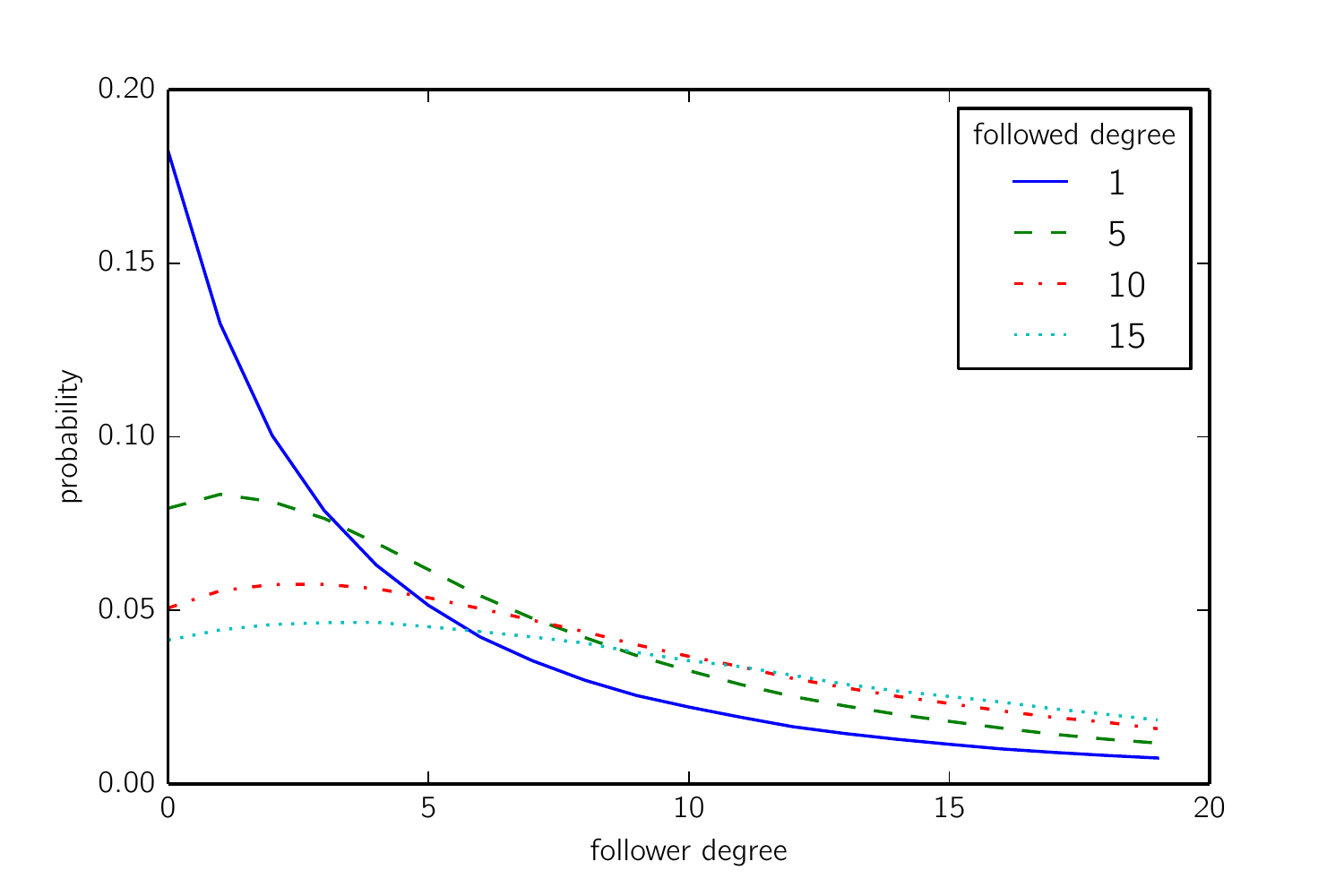}
    \caption{The degree distribution of the followers aggregated by the followed degree.
             Note the inflection point increases with a followed degree.}
    \label{fig:twitter_retweets_followers_degrees}
\end{figure}

Undoubtedly, the process of information transmission in social networks is far
more complex to be modeled by simulations just on the random power-law graphs. Because of
the hierarchical structure and the relation of activity with node degree effects, we propose
the model where the information is spread only to the nodes with lower degree
with approximately uniform distribution.

\subsection{Random DAG Generator}

One of the basic methods for generating random graphs has been introduced in
1960 by \erdosandrenyi~\cite{erdos}. In a nutshell: for a
given set of vertices, all edges have the same probability of being present or
absent in a graph. This model of the graph is not suited for modeling the social
networks because its degree distribution does not follow a power-law. The
distribution of degrees for \erdosrenyi model is binomial.

\begin{figure}[!ht]
    \centering
    \includegraphics[width=0.5\textwidth]{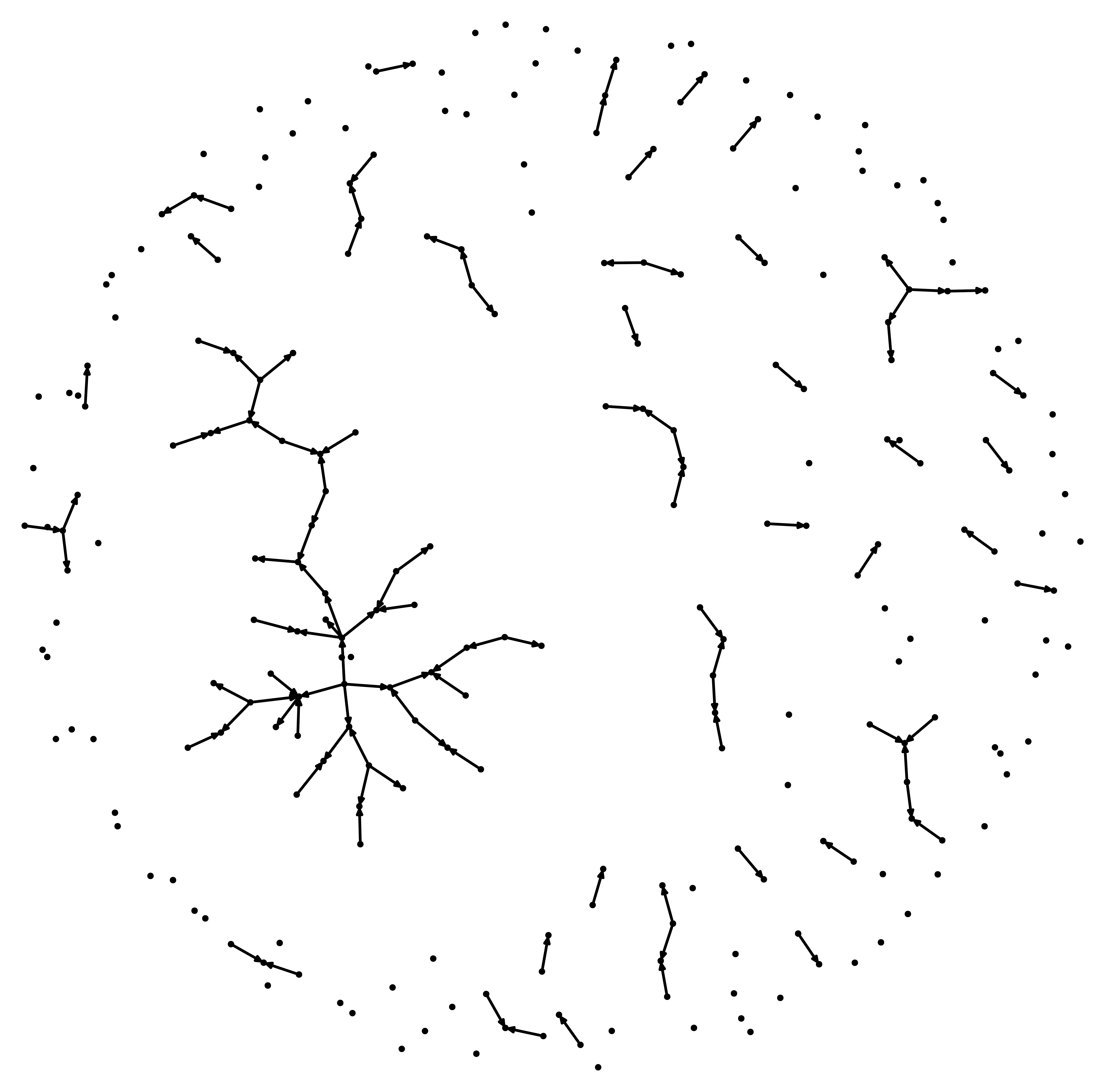}
    \caption{A diffusion network generated by the random DAG algorithm (see Algorithm~\ref{generate_dag}).}
    \label{fig:sample_graph}
\end{figure}

This model is used in the theoretical research for modeling interactions
between networks and propagation of catastrophes~\cite{catastrofic}. Even though
\erdosrenyi model does not characterize connections between nodes, we believe it
can represent the process of information spreading. We will propose the intuitive
variation of \erdosrenyi model for directed graphs.

According to \naivecitemany{Leskovec}{leskovec-blog} the cascades very rarely express cycles and
can be modeled as a tree-like structure. Even though this graphs do not model
relationships in the social network, the directed acyclic graphs (DAG) are
appropriate structure of the information propagation in the social network.

We introduce the procedure that generates the random directed acyclic graphs (DAG) and
prove that propagating information in CGM regime results in cascade sizes
obeying the power-law.

Let us denote by \rdag{n}{p} a random graph generated by the \textsc{RandomDAG($n,p$)}
 (see Algorithm~\ref{generate_dag}).

\begin{algorithm}
\caption{Generation of a random DAG}
\label{generate_dag}
\begin{algorithmic}
\Procedure{RandomDAG}{$n,p$}
\State $G \gets \text{empty graph with vertices } 1,2, \ldots, n$
\For{$i \gets 1 \text{ to } n-1$}
  \For{$j \gets i+1 \text{ to } n$}
    \State with probability $p$ add directed edge $(j,i)$ to $G$ 
  \EndFor
\EndFor
\State \textbf{return} $G$
\EndProcedure
\end{algorithmic}
\end{algorithm}

The final $n$-vertices graph is acyclic since all edges $(j,i)$ obey $j>i$ 
(see Figure~\ref{fig:sample_random_dag}).

Any DAG can be generated by the \textsc{RandomDAG($n,p$)}
when $p \in (0,1)$. If we label all vertices from $n$ to $1$ in the topological
order, the graph $G$ will consist only of edges $(j,i)$ that $j > i$.
Finally, all edges obeying $j > i$ can be present in the graph with independent
probability $p$.

\begin{figure}[!ht]
\centering
\begin{tikzpicture}[scale=1.0, every node/.style={draw,circle}]
  \node (1) at (-3, 0) { 1 };
  \node (2) at (-1.5, 0) { 2 };
  \node (3) at ( 0, 0) { 3 };
  \node (4) at ( 1.5, 0) { 4 };
  \node (5) at ( 3, 0) { 5 };

  \draw[<-, >=latex] (1) edge (2);
  \draw[<-, >=latex] (2) edge (3);
  \draw[<-, >=latex] (3) edge (4);
  \draw[<-, >=latex] (4) edge (5);

  \draw[<-, >=latex] (1) edge [bend left] (4);
  \draw[<-, >=latex] (2) edge [bend left] (4);
  \draw[<-, >=latex] (3) edge [bend left] (5);

  \draw[<-, >=stealth,dotted] (1) edge [bend right] (5);
  \draw[<-, >=stealth,dotted] (1) edge [bend right] (3);
  \draw[<-, >=stealth,dotted] (2) edge [bend right] (5);

\end{tikzpicture}
\caption{Sample DAG generated by procedure \textsc{RandomDAG($5,p$)}. Dotted edges are unchosen.}
\label{fig:sample_random_dag}
\end{figure}
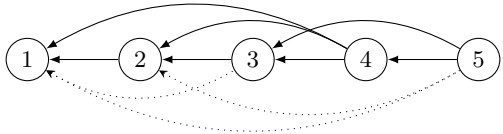

\naivecitemany{Leskovec}{leskovec-blog} suggested that in the real information
diffusion network, the small and simple graphs will occur more often than the
complex, non trivial DAGs. This is exactly the case in
\textsc{RandomDAG($n,p$)}.

\subsubsection{Degree Distribution of a Random DAG}

The distribution of in-degrees in the \rdag{n}{p} satisfies:

\begin{equation} \label{formula_rdag_degrees}
   \prob{\mathrm{indeg}(v) = k} =
   \frac{1}{n} \sum^{n-1}_{i=k} \binom{i}{k} p^k (1-p)^{i-k}
   .
\end{equation}

Similarly to the \erdosrenyi graph, the in-degree distribution of a given vertex $i$ is binomial
but with different parameters for each node. As a consequence, it leads to a 
step-function like shape (see Figure~\ref{fig:random_degree_dist}).

\begin{figure}[!ht]
    \centering
    \includegraphics[width=0.5\textwidth]{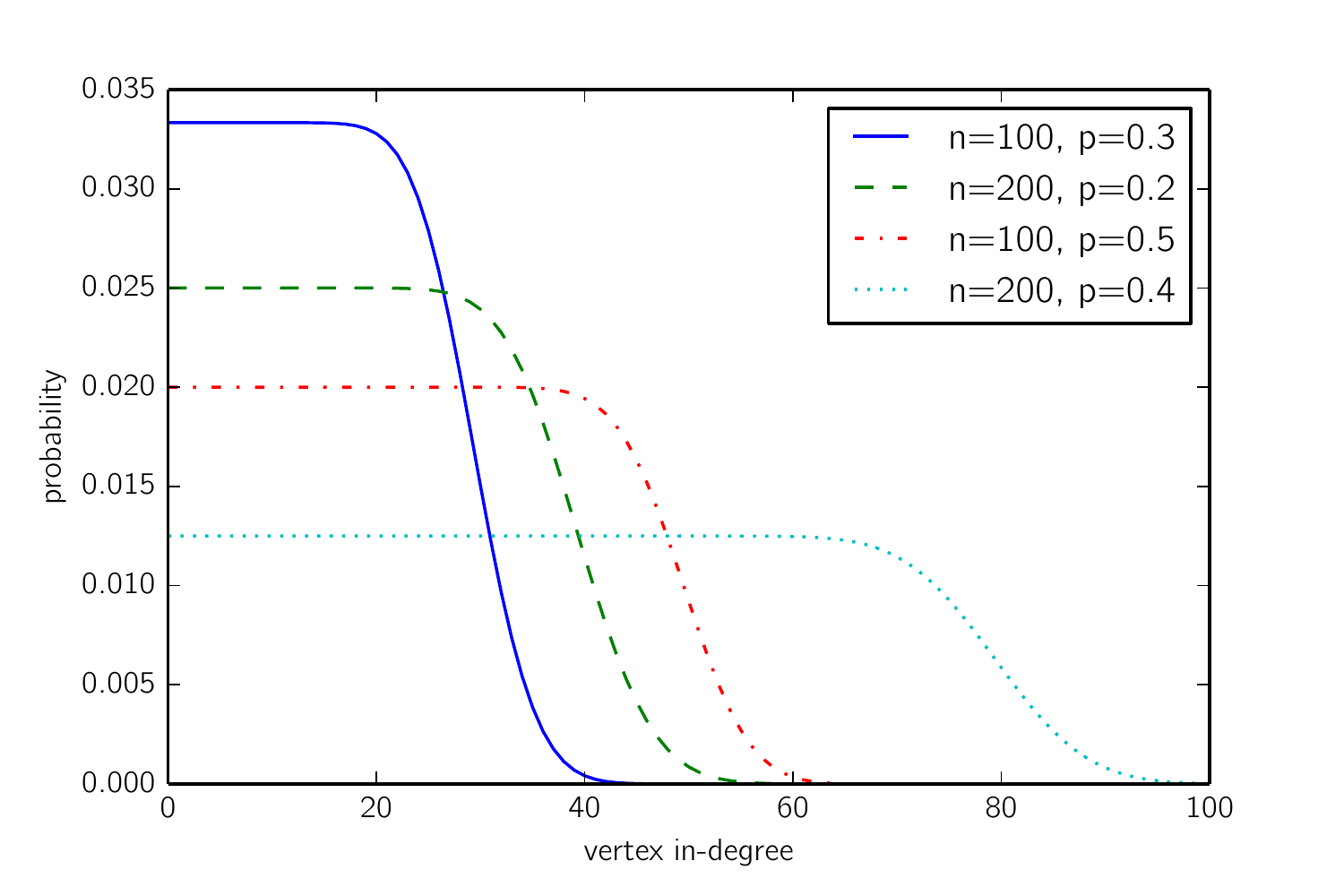}
    \caption{The in-degree distribution of the \rdag{n}{p} graphs, according
    to Formula~\eqref{formula_rdag_degrees}.}
    \label{fig:random_degree_dist}
\end{figure}

The in-degree distribution is almost uniform when the in-degree is lower than
$n p$. Naturally, this is merely a simple
approximation of the true network of a probable information transmission. Still,
based on observations in Section~\ref{information-transmission} \rdag{n}{k}
is the more accurate model than the follower-followee graph.

\subsubsection{Comparison with \erdosrenyi Model}

The random DAG, in contrast to the standard \erdosandrenyi~graph is a
directed graph. It means that it can model one-way communication and
the distribution of degrees. Next key difference 
is the hierarchical structure (i.e., the node $n$ will not follow
the node with the lower label). These differences enable the random DAG model to produce
cascades with the power-law distribution of sizes.

\section{Analysis} \label{analysis}

In this section, we will formally analyze the introduced model.
Subsequently, we will quantitatively describe a process of the information diffusion on a
random DAG by determining the cascade size distribution.

Recall $p$ to be the probability of an edge in a $\rdag{n}{p}$ and $\alpha$ to be the
average infectiousness of an information. Then set $\beta = 1-p\alpha$ for
simplicity. Hence, $\beta$ is the probability that the informed user will not
spread the information through a given edge.

Now we will determine the probability $P_{n,k}$, that the cascade reaches
$k$ distinct vertices in a graph with $n$ nodes, commencing from the vertex no. $1$.
Clearly $P_{1,1} = 1$ and $P_{n,k} = 0$ when $k > n$ or $k \le 0$. For a remaining case
assume the cascade has size $k$ and consider two distinct states of the $n$-th node:

\begin{itemize}
\item $n$ was \underline{informed}. Then, at least one of the other $k-1$ informed
    nodes had passed it to the $n$-th node, so the probability is $(1-\beta^{k-1}) \cdot P_{n-1,k-1}$.
\item $n$ was \underline{not informed}. It can happen only when none of the other $k$ informed
    nodes passed the information to $n$-th node. The probability of such event
    equals $\beta^k \cdot P_{n-1,k}$.
\end{itemize}

Hence, we obtain a formula when $1 \le k \le n$:

\begin{displaymath}
        P_{n,k}  = \beta^k \cdot P_{n-1,k} +  (1-\beta^{k-1}) \cdot P_{n-1,k-1},
\end{displaymath}


To determine the distribution of the cascade size, we assume that an information shall
commence in any node with an equal probability. Because the process of propagating
the information in $\rdag{n}{p}$ starting from node $1$ is identical to propagating it from node $i$ in
$\rdag{n+i-1}{p}$, the distribution is:

\begin{displaymath}
    \prob{|\emph{Informed}| = k} = S_{n,k}  = \frac{1}{n} \sum_{i=1}^{n} P_{i,k}.
\end{displaymath}

Now, we have the exact equation for the cascade size distribution. This equation
does not have a simple form. However, we can ask what will happen when the number of
nodes in graph is large.
Let us recall that two series $x_n, y_n$ are \emph{asymptotically equivalent}
when:
$$x_n \sim y_n \text{ iff } \lim_{n \to \infty}\frac{x_n}{y_n} = 1.$$

The cascade size distribution satisfies Theorem~\ref{asymptotic-cascade}.

\begin{theorem}
    \begin{displaymath} \label{formula_s_nk}
        S_{n,k} \sim \frac{1}{n(1-\beta^k)}
    \end{displaymath}
    \label{asymptotic-cascade}
\end{theorem}

\begin{proof}
Let us denote $\widetilde{S}_{n,k} = n S_{n,k}$. We need to prove
that:
\begin{equation}\label{proveit}
    A_k := \lim_{n \to \infty} \widetilde{S}_{n,k} = \frac{1}{1-\beta^k}
    .
\end{equation}

We will prove it by induction.
For $k=1$:

\begin{eqnarray*}
    \widetilde{S}_{n,1} & = & \sum_{i=1}^n P_{i,1} = P_{1,1} + \sum^{n-1}_{i=1} P_{i+1,1} 
         =  1 + \sum^{n-1}_{i=1} \beta P_{i,1} \\ 
        & = &  1 + \beta \widetilde{S}_{n-1,1} = 1 + \beta + \beta^2 + \ldots +
        \beta^{n-1}
    .
\end{eqnarray*}

Hence, $\widetilde{S}_{n,1}$ is the sum of the geometric series:

\begin{displaymath}
    \widetilde{S}_{n,1} = \frac{1-\beta^n}{1-\beta} \to \frac{1}{1-\beta}
    .
\end{displaymath}

For $k > 1$:
\begin{displaymath}
     \widetilde{S}_{n,k} = \sum_{i=1}^n P_{i,k} = \sum_{i=1}^n \beta^k P_{i-1,k} + \sum_{i=1}^n (1-\beta^{k-1}) P_{i-1,k-1}
     .
\end{displaymath}

So $\widetilde{S}_{n,k}$ obeys the recursive formula:

\begin{equation} \label{eq1}
    \widetilde{S}_{n,k} =  \beta^k \widetilde{S}_{n-1,k} + (1-\beta^{k-1})  \widetilde{S}_{n-1,k-1}
    .
\end{equation}

Technical induction shows, that the $\widetilde{S}_{n,k}$ is bounded and
increasing in respect to $n$, so  $A_k = \lim_{n
\to \infty} \widetilde{S}_{n,k}$ exists. Hence, we can take a limit on both
sides of Equation~\ref{eq1} and obtain:

\begin{displaymath}
    A_k = \beta^k A_k + (1-\beta^{k-1}) A_{k-1}
    ,
\end{displaymath}

\begin{equation}\label{Ak_recursive}
    A_k = \frac{1-\beta^{k-1}}{1-\beta^k} A_{k-1}
    .
\end{equation}

Finally, by unwinding the recursive Formula~\eqref{Ak_recursive},
for each $k>0$ we obtain (recall that $A_1 = \frac{1}{1-\beta}$):
\begin{displaymath}
    A_k = \frac{1-\beta^{k-1}}{1-\beta^k}
          \frac{1-\beta^{k-2}}{1-\beta^{k-1}}
          \cdots \frac{1}{1-\beta}
        = \frac{1}{1-\beta^k}.
\end{displaymath}

Hence, we have proved an asymptotic Formula~\eqref{proveit} of the cascade size distribution.
\end{proof}

\subsection{Approximation for the Large Networks}

Recall, that $\beta = 1 - p\alpha = 1 - \epsilon$. Because $p$ and $\alpha$ are extremely
small, $\beta$ is close to 1. Taking the Laurent series of
our function we get:

\begin{displaymath}
    \frac{1}{1-(1-\epsilon)^k} = \frac{1}{k\epsilon} + \frac{k-1}{2k} + O(\epsilon)
    .
\end{displaymath}

The social networks have an extremely large
number of nodes (e.g., the Twitter network has about $300$ millions distinct
users~\cite{brach}). On the other hand, new information reaches only few
nodes (the cascade size distribution is believed to be a power-law for $k$ smaller
than $10\,000$)~\cite{cikm2014,leskovec-blog}.  Then, because the element $\frac{k-1}{2 k n}$
is insignificant when $n$ is that large, for $k \ll \frac{1}{p\alpha} \ll n$ we
get:

\begin{displaymath}
    S_{n,k} \sim \frac{1}{n(1-(1-p\alpha)^k)} = \frac{1}{k n p \alpha} +
    \frac{k-1}{2k n} + O(\frac{p \alpha}{n})
\end{displaymath}

Hence, the distribution of cascade size:

\begin{equation}\label{exact-formula}
    \prob{|\mathrm{Informed}| = k} \approx \frac{1}{n p \alpha } k^{-1} + \mathrm{const}
\end{equation}

in the first-order perturbation follows the power-law. Due to the low
number of the large cascades, the distribution of sizes is unknown for $k$ close to
$\frac{1}{p\alpha}$. In that case, one should use an exact formula.

\subsection{Goodness of the Approximation}

On the Figure~\ref{fig:analysis} we have presented a comparison between
approximation and exact formula for the distribution of cascade sizes. For a relatively small cascade size
$k$ the slope of a distribution matches ideally.

\begin{figure}[!ht]
    \centering
    \includegraphics[width=0.5\textwidth]{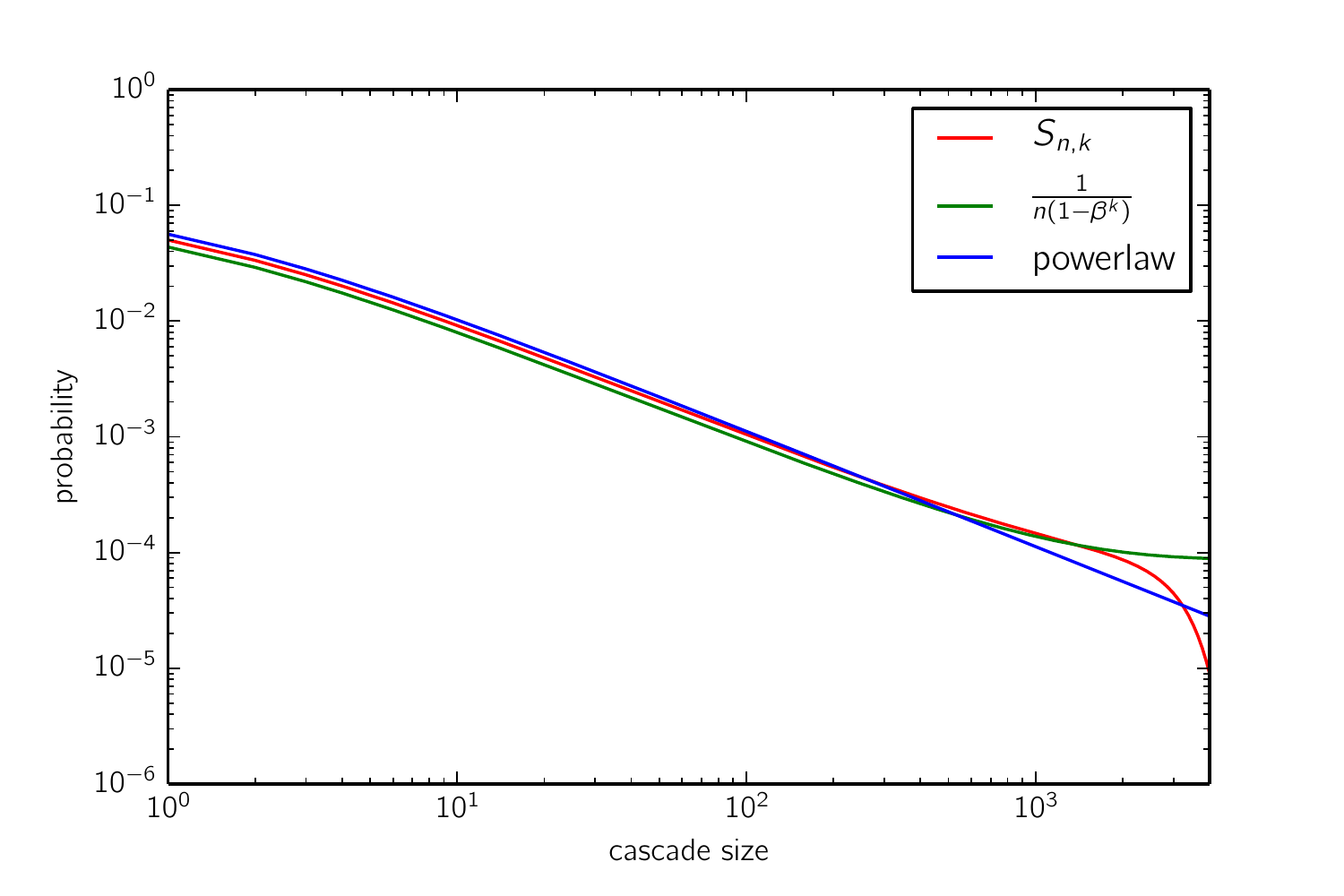}
    \caption{The log-log plot of an exact formula $S_{n,k}$,
        asymptotic bound $\frac{1}{n(1-\beta^k)}$ and the power-law distribution
        with exponent $-1$.
        }
     \label{fig:analysis}
\end{figure}

How large cascades can we model using the aforementioned assumptions? The number of
nodes $n$ in the Twitter network is approximately 300 millions. In~\cite{seismic}
the average infectiousness $\alpha$ of the information on Twitter is said to be
of order of $0.01$.
According to \naivecitemany{Leskovec}{leskovec-blog}, the number of edges in a cascade
is proportional to $n^{1.03}$. So the parameter $p$ should be
approximately $p \propto \frac{n^{1.03}}{n^2}$, since the number of possible
edges is $n^2$. The largest rumor in our dataset has roughly $70\,000$
informed nodes, hence the lower bound for parameter $p$ is of order of
$\frac{(7\cdot10^4)^{1.03}}{(7\cdot10^4)^2} \approx 2\cdot10^{-5}$.

Still, $\frac{1}{\alpha p} \approx 5 \cdot 10^6 \ll 3 \cdot 10^8 \approx n$.
So, for the Twitter network we can model the cascades with sizes $k \ll 5 \cdot 10^6$.
Remarkably, it is enough, since approximately $10^{-5}$ of Twitter rumors have the size
greater than $10^4$.

\subsubsection{Goodness of Fit}

K-S test comparing the power-law distribution and the real cascade size
distribution~\cite{twitter-data} is $0.0145$. At~\cite{hypertext2016}, the the
K-S test comparing the cascade size distribution on the variant of CGM with real
cascade distribution is $0.0447$. It means, that our model can potentially
improve the test value by $3\%$ in comparison to the CGM on the graph of
retweets.

\subsection{Other Schemes of Information Diffusion}

One would state that \textit{cascade generation model} is counterintuitive for
microblogging services such as the Twitter. In fact, it is more intuitive that
every follower of the spreader eventually will be informed. So every follower, after being informed for the first time
ought to make exactly one decision (with probability $\alpha$):
whether to pass the information to all of its acquaintances simultaneously (previously
the information was passed to each of its followers independently
and each follower may had multiple opportunities to become a spreader). In such
a case, the Formula~\eqref{exact-formula} is exactly the same (for further proof see
Appendix~\ref{model-many}).

\section{Conclusion and Future Work} \label{conclusion}

The graph of the information diffusion is utterly different from the global network
of social connections. In contrast to multiple previous approaches, we model the
cascade of information propagation as the \textit{random directed acyclic graph} and
we show that in the scheme of CGM the distribution of information popularity is asymptotically equivalent to:

\begin{displaymath}
    \prob{|\textit{Informed}| = k} \sim \frac{1}{n(1-\beta^k)}
    ,
\end{displaymath}

where $n$ is the number of nodes and $\beta$ is a parameter dependent on both
infectiousness of average information and density of the cascade. We show that for
a sufficiently big number of nodes this distribution follows the power-law,
 what is consistent with real world observations. We hope that an introduction of this framework will inspire
the theoretical affords to model and describe the information diffusion in
the large social networks. The analysis of the \rdag{n}{p} graph showed that
the cascade size distribution follows the power-law $\prob{|Informed|=k} \propto
k^\gamma$ for $\gamma=-1$. \naivecitemany{Leskovec}{leskovec-dvd} presented the rich family of data with power-law exponent
close to $-1$. \naivecitemany{Leskovec}{leskovec-dvd} suggest that the information
propagation in the network of recommendations may produce cascades with desired
exponent. 
However, in the real data, the parameter $\gamma$ can be completely different (e.g.,
see Figure~\ref{fig:twitter_reach_distribution} of the Twitter cascade size
distribution with power-law exponent $\gamma=-2.3$).
To adopt our \rdag{n}{p} model to those cases, one can customize a distribution
of random cascades (we assumed a fairly simple method to generate
them).
We believe that adaptation of cascade shape distribution
to the experimental data will readjust the $\gamma$ parameter (for a start one
can use the shape distribution provided by~\naivecitemany{Leskovec}{leskovec-blog}). Further 
enhancements might also be achieved by adapting the information diffusion scheme to
a particular society (similarly to Appendix~\ref{model-many}).
To encourage other researchers to apply our model in practice, we publish the
code used to generate all figures and the results on~\cite{rdag-code}.

Still, we need more experiments and research to answer what type of social networks the
random directed acyclic graphs model and how richer set of features (e.g.,
spatio-temporal features) influences the cascade size distribution. 

\section{Acknowledgments}

This work was partially supported by NCN grant UMO-2014/13/B/ST6/01811, ERC
project PAAl-POC 680912, ERC project TOTAL 677651, FET IP project MULTIPLEX
317532 and polish funds for years 2013-2016 for co-financed international
projects.

\bibliographystyle{abbrv}
\bibliography{sigproc} 

\appendix

\section{Dependant Passing of the Information}
\label{model-many}
In this model we count spreaders who pass the information to all of theirs followers with probability $\alpha$.
These followers will receive the information and might become new spreaders.

For clarity, we will use the notation from Section~\ref{analysis} and the proof of Theorem~\ref{formula_s_nk}.

Analogously to previous analysis, we obtain the recursive formula:
\begin{eqnarray*}
    P_{1,1} &=& \alpha, \\
    P_{n,k} &=& 0, \text{ when } k > n, \\
    P_{n,k} &=& \big(1-(1-\beta^k)\alpha\big) \cdot P_{n-1,k} + \alpha (1 - \beta^{k-1})
    \cdot P_{n-1,k-1} .
\end{eqnarray*}

Rest of the proof is almost identical to the proof of Theorem~\ref{formula_s_nk}.

For $k=1$, we have:
\begin{displaymath}
    \lim_{n \to \infty} \widetilde{S}_{n,1} = 
    \lim_{n \to \infty} \frac{1-\big(1-(1-\beta^k)\alpha\big)^n}{1-\beta} =
    \frac{1}{1-\beta}
    .
\end{displaymath}

For $k>1$, when $n \to \infty$:
\begin{eqnarray*}
    \lim_{n \to \infty} \widetilde{S}_{n,k} = A_k
    & = & \big(1-(1-\beta^k)\alpha\big)
          A_k + \\
    &   & + \alpha (1-\beta^{k-1}) A_{k-1}
    .
\end{eqnarray*}

By subtracting the expression on both sides:

\begin{displaymath} 
    A_k - A_k (1 - (1-\beta^k)\alpha) = A_{k-1} \alpha (1-\beta^{k-1})
    .
\end{displaymath}

And after simplification we get:

\begin{displaymath} 
    A_k (1-\beta^k)\alpha = A_{k-1} (1-\beta^{k-1}) \alpha
\end{displaymath}

Hence, we have obtained the same formula as in Theorem~\ref{formula_s_nk}:
\begin{displaymath} 
    A_k  = A_{k-1} \frac{1-\beta^k}{1-\beta^{k-1}}
\end{displaymath}

and finally obtain:
\begin{displaymath}
    \lim_{n\to\infty} \widetilde{S}_{n,k} = \frac{1}{1-\beta^k}
    .
\end{displaymath}

\end{document}